\documentclass[11pt]{article}
\title{Upper bounds for centerlines\footnote{This paper (without the appendix) has been published
in Journal of Computational Geometry 3:20--30, 2012.}}
\author{Boris Bukh\footnote{Centre for Mathematical Sciences, Cambridge CB3 0WB, England;
and Churchill College, Cambridge CB3 0DS, England.
\texttt{B.Bukh@dpmms.cam.ac.uk}.}
\and Gabriel Nivasch\footnote{Mathematics Department, EPFL, Lausanne, Switzerland. \texttt{gabriel.nivasch@epfl.ch}.
Part of this work was done when the author was at ETH Z\"urich, Z\"urich, Switzerland.}}
\date{}

\usepackage{amsmath,amsfonts,amssymb,amsthm}
\usepackage{hyperref}
\usepackage{graphics}

\newtheorem{theorem}{Theorem}[section]
\newtheorem{lemma}[theorem]{Lemma}

\newtheorem{claim}[theorem]{Claim}

\newenvironment{definition}[1][Definition]
{\refstepcounter{theorem}\begin{trivlist}\item\textbf{#1
\thetheorem:\ }} {\end{trivlist}}

\newenvironment{remark}[1][Remark]
{\refstepcounter{theorem}\begin{trivlist}\item\textbf{#1
\thetheorem:\ }} {\end{trivlist}}

\newcommand{\R}{\mathbb{R}}  
\newcommand{\Gs}{G_{\mathrm s}} 
\newcommand{\I}{[0,1]} 
\newcommand{\B}{\mathrm{BB}}    
\newcommand{\Ha}{\mathcal H'_{\mathrm a}}
\newcommand{\Hb}{\mathcal H'_{\mathrm b}}
\newcommand{\Hc}{\mathcal H'_{\mathrm c}}
\newcommand{\Hup}{\mathcal H'_{\mathrm{up}}}
\newcommand{\Hdown}{\mathcal H'_{\mathrm{down}}}
\DeclareMathOperator{\vol}{\mathrm{vol}}   

\begin{document}

\maketitle

\begin{abstract}
In 2008, Bukh, Matou\v{s}ek, and Nivasch conjectured that for
every $n$-point set $S$ in $\R^d$ and every $k$, $0 \le k \le
d-1$, there exists a $k$-flat $f$ in $\R^d$ (a ``centerflat")
that lies at ``depth" $(k+1) n / (k+d+1) - O(1)$ in $S$, in the
sense that every halfspace that contains $f$ contains at least
that many points of $S$. This claim is true and tight for $k=0$
(this is Rado's centerpoint theorem), as well as for $k = d-1$
(trivial). Bukh et al.~showed the existence of a $(d-2)$-flat
at depth $(d-1) n / (2d-1) - O(1)$ (the case $k = d-2$).

In this paper we concentrate on the case $k=1$ (the case of
``centerlines"), in which the conjectured value for the leading
constant is $2/(d+2)$. We prove that $2/(d+2)$ is an
\emph{upper bound} for the leading constant. Specifically, we
show that for every fixed $d$ and every $n$ there exists an
$n$-point set in $\R^d$ for which no line in $\R^d$ lies at
depth greater than $2n/(d+2) + o(n)$. This point set is the
``stretched grid"---a set which has been previously used by
Bukh et al.~for other related purposes.

Hence, in particular, the conjecture is now settled for $\R^3$.

\end{abstract}

\section{Introduction}

Given a finite set $S\subset \R^d$ and a point $x\in\R^d$,
define the \emph{depth} of $x$ in $S$ as the minimum of
$|\gamma \cap S|$ over all closed halfspaces $\gamma$ that
contain $x$. Rado's centerpoint theorem (1947,~\cite{rado})
states that for every $n$-point set $S\subset \R^d$ there
exists a point $x\in\R^d$ at depth at least $n/(d+1)$ in $S$.
Such a point $x$ is called a \emph{centerpoint}.

Centerpoints, besides being a basic notion in discrete
geometry, have also been studied in connection with statistical
data analysis: The centerpoint $x$ is a single point that
describes, in some sense, a given ``data set" $S$
\cite{chan_tukeydepth,dg_halfspacedepth,small_medians}.

The notion of \emph{depth} that we use in this paper is
sometimes called \emph{halfspace depth} or \emph{Tukey depth},
to distinguish it from other notions of depth (see, for
example, \cite{nabil}).

The constant $1/(d+1)$ in the centerpoint theorem is easily
shown to be tight: Take $d+1$ affinely independent points in
$\R^d$, and let $S$ be obtained by replacing each of these
points by a tiny ``cloud" of $n/(d+1)$ points. Then no point in
$\R^d$ lies at depth greater than $n/(d+1)$ in $S$.

In this paper we consider a generalization of the centerpoint
theorem in which the desired object is not a deep point, but
rather a deep $k$-flat for some given $0\le k<d$. Thus, let us
define the \emph{depth} of a $k$-flat $f\subset \R^d$ in $S$ as
the minimum of $|\gamma \cap S|$ over all closed halfspaces
$\gamma$ that contain $f$.

Bukh, Matou\v{s}ek, and Nivasch~\cite{BMN_stabbing} proved that
for every $n$-point set $S\subset\R^d$ there exists a
$(d-2)$-flat $f\subset \R^d$ at depth at least $(d-1)n/(2d-1) -
O(1)$.\footnote{They showed that there exist $2d-1$ hyperplanes
passing through a common $(d-2)$-flat that partition $S$ into
$4d-2$ parts, each of size at least $n/(4d-2) - O(1)$. This
$(d-2)$-flat is the desired $f$, since every halfspace that
contains it must completely contain at least $2d-2$ of the
parts.}

It is trivial that there always exists a $(d-1)$-flat at depth
at least $n/2$ in $S$. In~\cite{BMN_stabbing} it was
conjectured that, in general, for every $k$, $0\le k\le d-1$,
there exists a $k$-flat at depth at least $(k+1)n/(k+d+1) -
O(1)$ in $S$, and that the fraction $(k+1)/(k+d+1)$ is sharp.
Such a flat would be called a \emph{centerflat}; and we call
this conjecture the \emph{centerflat conjecture}.

The centerflat conjecture is closely related to the
\emph{center transversal theorem} of Dol'nikov~\cite{dolnikov}
and \v Zivaljevi\'c and Vre\'cica~\cite{zv}; it states that, if
$S_1, \ldots, S_{k+1} \subset \R^d$ are point sets of sizes
$n_1, \ldots, n_{k+1}$, respectively, then there exists a
$k$-flat $f\subset \R^d$ that simultaneously lies at depth at
least $n_i/(d-k+1)$ in each $S_i$.

As far as we know, however, the centerflat conjecture itself
has not been studied until very recently. Arocha et
al.~\cite{ABMR_transv} have obtained a lower bound of
$1/(d-k+1)$ for the leading constant in the cojecture (see
Corollary 3 there). The same constant can also be obtained
using the center transversal theorem, by setting all the sets
$S_i$ to $S$. Actually, one can obtain this constant much more
simply, by projecting $S$ into $\R^{d-k}$ and then applying the
centerpoint theorem. However, the conjectured constant of
$(k+1)/(k+d+1)$ is larger than $1/(d-k+1)$ for all $1\le k \le
d-2$.

In this paper we focus on the case $k=1$ of the centerflat
conjecture (the case of ``centerlines"). For this case the
conjecture predicts a value of $2/(d+2)$ for the leading
constant, and we show that this value cannot be improved.
Specifically:

\begin{theorem}\label{thm_2d2_tight}
Let $d\ge 2$ be fixed. Then, for every $n$ there exists an
$n$-point set $\Gs \subset \R^d$ such that for every line
$\ell\subset \R^d$ there exists a halfspace containing $\ell$
and containing at most $2n/(d+2) + o(n)$ points of $\Gs$.
\end{theorem}

Combining Theorem~\ref{thm_2d2_tight} with the above-mentioned
result in~\cite{BMN_stabbing}, we conclude that in $\R^3$ there
is always a line at depth $2n/5 - O(1)$, and that the fraction
$2/5$ is sharp.

The set $\Gs$ in the theorem is the ``stretched grid"---a point
set that was previously used by Bukh et al.~\cite{BMN_stair}
for obtaining lower bounds for weak $\epsilon$-nets, upper
bounds for the so-called \emph{first selection lemma}, and for
other related purposes (see also~\cite{GN_thesis}).
Unfortunately, we have been unable to find a simple,
``cloud"-based construction for proving
Theorem~\ref{thm_2d2_tight}, like the construction mentioned
above for centerpoints.

\section{The stretched grid and stair-convexity}

The \emph{stretched grid} is an axis-parallel grid of points
where, in each direction $i$, $2\le i \le d$, the spacing
between consecutive ``layers" increases rapidly, and
furthermore, the rate of increase for direction $i$ is much
larger than that for direction $i-1$. To simplify calculations,
we will also make the coordinates increase rapidly in the first
direction.\footnote{The most natural way to define the
stretched grid is using the notion of \emph{infinitesimals}
from nonstandard analysis. But we avoid doing so in order to
keep the exposition accessible.}

The formal definition is as follows: Given $n$, the desired
number of points, let $m = n^{1/d}$ be the side of the grid
(assume for simplicity that this quantity is an integer), and
let
\begin{equation}\label{eq_Gs}
\Gs = \bigl\{(K_1^{a_1}, K_2^{a_2}, \ldots, K_d^{a_d}) : a_i
\in \{0,\ldots, m-1\} \text{ for all $1\le i\le d$} \bigr\},
\end{equation}
for some appropriately chosen constants $1<K_1 \ll K_2 \ll K_3
\ll \cdots \ll K_d$. Each constant $K_i$ must be chosen
appropriately large in terms of $K_{i-1}$ and in terms of $m$.
We choose the constants as follows:
\begin{equation}\label{eq_Ki}
K_1 = 2d, \quad K_2 = K_1^m, \quad K_3 = K_2^m, \quad \ldots, \quad K_d = K_{d-1}^m.
\end{equation}

Throughout this paper we refer to the $d$-th coordinate as the
``height", so a hyperplane in $\R^d$ is \emph{horizontal} if
all its points have the same last coordinate; and a line in
$\R^d$ is \emph{vertical} if all its points share the first
$d-1$ coordinates. A \emph{vertical projection onto $\R^{d-1}$}
is obtained by removing the last coordinate. The \emph{$i$-th
horizontal layer of $\Gs$} is the subset of $\Gs$ obtained by
letting $a_d = i$ in (\ref{eq_Gs}).

The following lemma is not actually used in the paper, but it
provides the motivation for the stretched grid:

\begin{lemma}\label{lemma_property_Gs}
Let $a \in \Gs$ be a point at horizontal layer $0$, and let
$b\in \Gs$ be a point at horizontal layer $i$. Let $c$ be the
point of intersection between segment $ab$ and the horizontal
hyperplane containing layer $i-1$. Then $|c_i - a_i| \le 1$ for
every $1\le i\le d-1$.
\end{lemma}

Lemma~\ref{lemma_property_Gs} follows from a simple calculation
(we chose the constants $K_i$ in (\ref{eq_Ki}) large enough to
make this and later calculations work out).

The grid $\Gs$ is hard to visualize, so we apply to it a
logarithmic mapping $\pi$ that converts $\Gs$ into the uniform
grid in the unit cube.

Formally, let $\B = [1,K_1^{m-1}] \times \cdots \times [1,
K_d^{m-1}]$ be the bounding box of the stretched grid, let
$[0,1]^d$ be the unit cube in $\R^d$, and define the mapping
$\pi \colon \B \to [0,1]^d$ by
\begin{equation*}
\pi(x) = \left( \frac{\log_{K_1} x_1}{m-1}, \ldots, \frac{\log_{K_d} x_d}
{m-1} \right).
\end{equation*}

Then, it is clear that $\pi(\Gs)$ is the uniform grid in
$[0,1]^d$.

We say that two points $a, b\in\B$ are \emph{$c$-close in
coordinate $i$} if the $i$-th coordinates of $\pi(a)$ and
$\pi(b)$ differ by at most $c/(m-1)$. Roughly speaking, this
means that $a$ and $b$ are at most $i$ layers apart in the
$i$-th direction. Otherwise, we say that $a$ and $b$ are
\emph{$c$-far in coordinate $i$}. Two points are
\emph{$c$-close} if they are $c$-close in every coordinate, and
they are \emph{$c$-far} if they are $c$-far in every
coordinate.

\begin{figure}
\centerline{\includegraphics{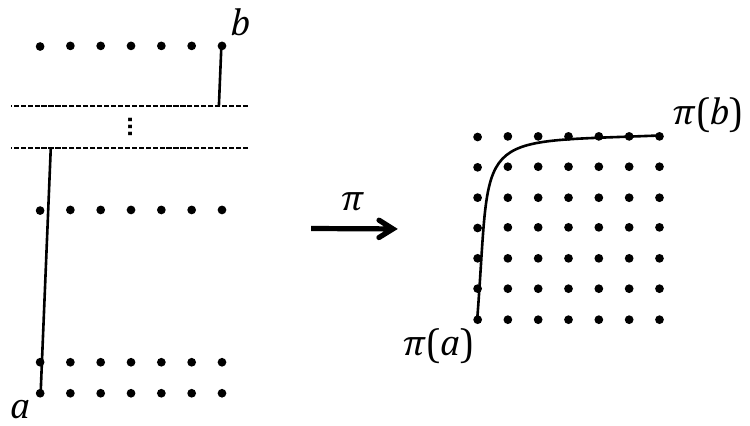}}
\caption{\label{fig_Gs_pi}The stretched grid and the mapping $\pi$ in the plane.
The stretched grid is too tall to be drawn entirely, so an intermediate portion of
it has been omitted. A line segment connecting two points is also shown, as well as
its image under $\pi$. (The first coordinate of the stretched grid does not increase geometrically
in this picture.)}
\end{figure}

Lemma~\ref{lemma_property_Gs} implies that the map $\pi$
transforms straight-line segments into curves composed of
almost-straight axis-parallel parts: Let $s$ be a straight-line
segment connecting two points of $\Gs$. Then $\pi(s)$ ascends
almost vertically from the lower endpoint, almost reaching the
height of the higher endpoint, before moving significantly in
any other direction; from there, it proceeds by induction. See
Figure~\ref{fig_Gs_pi}.

This observation motivates the notions of
\emph{stair-convexity}, which describe, in a sense, the limit
behavior of $\pi$ as $m\to\infty$.

\subsection{Stair-convexity}

We recall a few notions from \cite{BMN_stair}.

Given a pair of points $a, b\in \R^d$, the \emph{stair-path}
$\sigma(a,b)$ between them is a polygonal path connecting $a$
and $b$ and consisting of at most $d$ closed line segments,
each parallel to one of the coordinate axes. The definition
goes by induction on $d$; for $d=1$, $\sigma(a,b)$ is simply
the segment $ab$. For $d\ge 2$, after possibly interchanging
$a$ and $b$, let us assume $a_d\le b_d$. We set $a' = (a_1,
\ldots, a_{d-1}, b_d)$, and we let $\sigma(a,b)$ be the union
of the segment $aa'$ and the stair-path $\sigma(a', b)$; for
the latter we use the recursive definition, ignoring the common
last coordinate of $a'$ and $b$.

Note that, if $c$ and $d$ are points along $\sigma(a, b)$, then
$\sigma(c, d)$ coincides with the portion of $\sigma(a, b)$
that lies between $c$ and $d$.

A set $X\subseteq \R^d$ is said to be \emph{stair-convex} if
for every $a, b\in X$ we have $\sigma(a, b) \subseteq X$.

Given a set $X\subset \R^d$ and a real number $h$, let $X(h)$
(the \emph{horizontal slice at height $h$}) be the vertical
projection of $\{x \in X : x_d = h \}$ into $\R^{d-1}$.
In~\cite{BMN_stair} it was shown that a set $X\subset \R^d$ is
stair-convex if and only if the following two conditions hold:
(1) every horizontal slice $X(h)$ is stair-convex; (2) for
every $h_1 \le h_2 \le h_3$ such that $X(h_3) \neq\emptyset$ we
have $X(h_1) \subseteq X(h_2)$ (meaning, the horizontal slice
can only grow with increasing height, except that it can end by
disappearing abruptly).\footnote{This criterion was stated
slightly incorrectly in \cite{BMN_stair}; the formulation given
above is the correct one.} For convenience we call this
criterion \emph{monotonicity of slices}.

Let $a\in \R^d$ be a fixed point, and let $b\in \R^d$ be
another point. We say that $b$ has \emph{type $0$ with respect
to $a$} if $b_i\le a_i$ for every $1\le i\le d$. For $1\le j\le
d$ we say that $b$ has \emph{type $j$ with respect to $a$} if
$b_j\ge a_j$ but $b_i\le a_i$ for every $i$ satisfying $j+1\le
i\le d$. (It might happen that $b$ has more than one type with
respect to $a$, but only if some of the above inequalities are
equalities.)

Given a point $a\in\R^d$, let $C_i(a)$ (the \emph{$i$-th
component with respect to $a$}) be the set of all points in
$\R^d$ that have type $i$ with respect to $a$. Thus,
\begin{align*}
C_0(a) &= (-\infty, a_1] \times \cdots \times (-\infty, a_d];\\
C_i(a) &= (-\infty, \infty)^{i-1} \times [a_i, \infty) \times (-\infty, a_{i+1}] \times
\cdots \times (-\infty, a_d], \qquad \text{for $1\le i \le d$}.
\end{align*}
See Figure~\ref{fig_components}.

\begin{figure}
\centerline{\includegraphics{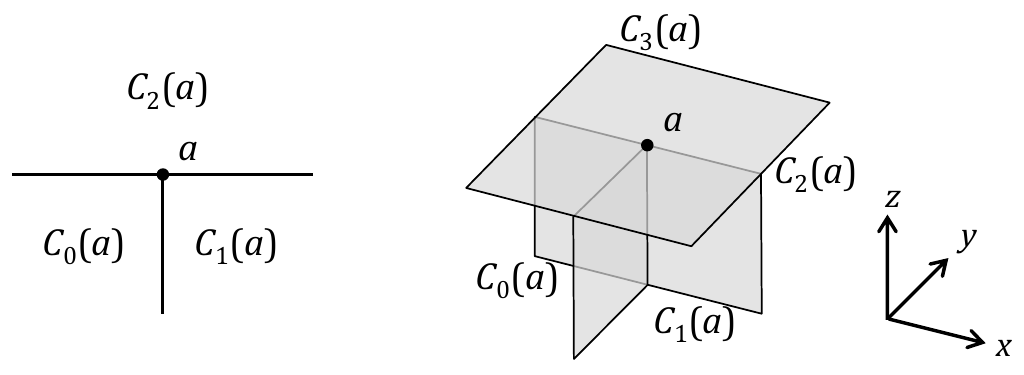}}
\caption{\label{fig_components}Components with respect to a point $a$ in the plane (left) and in $\R^3$ (right).}
\end{figure}

We now introduce a new notion, that of a
\emph{stair-halfspace}. Stair-halfspaces are, roughly speaking,
the stair-convex analogue of Euclidean halfspaces.

\begin{definition}\label{def_stair-halfspace}
Let $a\in\R^d$ be a point, and let $\emptyset \subsetneq I
\subsetneq \{0,\ldots,d\}$ be a set of indices. Then the set
$\bigcup_{i\in I} C_i(a)$ is called a \emph{stair-halfspace},
and $a$ is its \emph{vertex}.
\end{definition}

\begin{lemma}\label{lemma_stair-halfspace_stconv}
Let $H$ be a stair-halfspace. Then both $H$ and $\R\setminus H$
are stair-convex.
\end{lemma}

\begin{proof}
Consider $H$. Every horizontal slice $H(h)$ of $H$ is either
empty, all of $\R^{d-1}$, or a $(d-1)$-dimensional
stair-halfspace. Thus, by induction, $H(h)$ is always
stair-convex. Furthermore, for every $h_1\le h_2 \le h_3$ such
that $H(h_3)\neq\emptyset$ we have $H(h_1) \subseteq H(h_2)$.
Thus, $H$ is stair-convex by monotonicity of slices. A similar
argument applies for $\R\setminus H$.
\end{proof}

(There are other sets in $\R^d$ that deserve to be called
\emph{stair-halfspaces}, that do not fit into the above
definition; for example, the set $\{(x,y) \in \R^2 : x\ge 0\}$.
But Definition~\ref{def_stair-halfspace} covers all the
stair-halfspaces that we will need in this paper.)

Two stair-halfspaces $\bigcup_{i\in I} C_i(a)$ and
$\bigcup_{i\in I} C_i(b)$ with the same index set $I$ are said
to be \emph{combinatorially equivalent}.

Note that the map $\pi$ preserves stair-convexity notions
(since it operates componentwise and is monotone in each
component). In particular, let $X\subseteq \B$; then: (1) $X$
is stair-convex if and only if $\pi(X)$ is stair-convex; (2)
$X$ is a stair-path if and only if $\pi(X)$ is a stair-path;
(3) there exists a stair-halfspace $H$ such that $X = H \cap
\B$ if and only if there exists a combinatorially equivalent
stair-halfspace $H'$ such that $\pi(X) = H' \cap [0,1]^d$.

The following lemma shows that every stair-halfspace is, in a
sense, the limit of the image under $\pi$ of a Euclidean
halfspace.

\begin{lemma}\label{lemma_stair-half_to_half}
Let $a\in BB$ be a point, and let $H = \bigcup_{i\in I} C_i(a)$
be a stair-halfspace with vertex $a$ and index set $\emptyset
\subsetneq I\subsetneq \{0,\ldots, d\}$. Then there exists a
Euclidean halfspace $H'$ with $a\in\partial H'$ such that, for
every point $x\in\B$ that is $1$-far from $a$, we have $x\in H$
if and only if $x\in H'$. (See
Figure~\ref{fig_approx_halfspace}.)
\end{lemma}

\begin{figure}
\centerline{\includegraphics{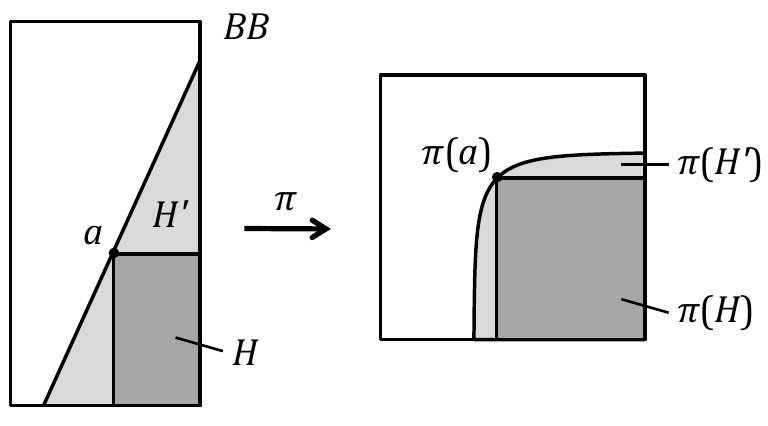}}
\caption{\label{fig_approx_halfspace}For every stair-halfspace $H$ there exists a Euclidean halfspace $H'$ that closely approximates $H$ within $\B$: The images of $H$ and $H'$ under $\pi$ almost coincide. (The figure is not to scale.)}
\end{figure}

\begin{proof}
The desired Euclidean halfspace is
\begin{equation*}
H' = \{ x \in \R^d : s_0 + s_1\frac{x_1}{a_1} + s_2\frac{x_2}{a_2} + \cdots +
s_d\frac{x_d}{a_d} \ge 0\},
\end{equation*}
where the $s_i$'s are small signed integers chosen to satisfy
the following conditions:
\begin{enumerate}
\item For every $0\le i\le d$, $s_i$ is positive if $i\in
    I$, and negative otherwise.

\item $\sum_i s_i = 0$.

\item We have $1\le |s_i|\le d$ for all $i$.
\end{enumerate}
Such a choice is always possible since $1\le|I|\le d$, so there
will be both positive and negative $s_i$'s.

The second condition above ensures that $a$ lies on the
boundary of $H'$.

Now consider a point $x\in\B$ that is $1$-far from $a$. Thus,
we have either $x_i \ge K_ia_i$ or $x_i\le a_i/K_i$ for every
coordinate $i$.

Let $i$ be the largest coordinate such that $x_i\ge K_ia_i$, if
it exists. Consider the sum
\begin{equation}\label{eq_sum_halfspace}
s_0 + s_1\frac{x_1}{a_1} + s_2\frac{x_2}{a_2} + \cdots + s_d\frac{x_d}{a_d}.
\end{equation}
We claim that the term $s_i\frac{x_i}{a_i}$ is larger in
absolute value than all the other terms in
(\ref{eq_sum_halfspace}) combined. Indeed, for $j<i$ we have
\begin{equation*}
|s_j| \frac{x_j}{a_j} \le |s_j|x_j \le |s_j| K^{m-1}_j = |s_j| \frac{K_{j+1}}{K_j} \le |s_j| \frac{K_i}{K_j} \le |s_j|\frac{x_i}{a_iK_j} \le d|s_i|\frac{x_i}{a_iK_j} \le 2^{-j} |s_i| \frac{x_i}{a_i}.
\end{equation*}
This is because the constants $K_i$ were chosen appropriately
large in (\ref{eq_Ki}). Similarly, for $j>i$ we have $|s_j|
\frac{x_j}{a_j} \le d/K_j \le 2^{-j}$.

Thus, the sign of (\ref{eq_sum_halfspace}) is the sign of
$s_i$, which implies that $x\in H'$ if and only if $i\in I$.

If, on the other hand, $x_i\le a_i/K_i$ for all $i$, then, by a
similar argument, the sign of (\ref{eq_sum_halfspace}) is the
sign of $s_0$, so $x\in H'$ if and only if $0\in I$.
\end{proof}

The following lemma formalizes what we mean by translating a
stair-halfspace ``outwards":

\begin{lemma}\label{lemma_translate_st_halfsp}
Let $H = \bigcup_{i\in I} C_i(a)$ be a stair-halfspace with
vertex $a\in \R^d$ and index set $\emptyset\subsetneq
I\subsetneq \{0,\ldots,d\}$. Let $b\in\R^d$ be another point
such that, for each $1\le i\le d$, we have $b_i<a_i$ if $i\in
I$, and $b_i>a_i$ otherwise.

Let $H'=\bigcup_{i\in I} C_i(b)$ be the stair-halfspace
combinatorially equivalent to $H$ with vertex $b$. Then $H
\subset H'$. (See Figure~\ref{fig_translate_st_halfsp}.)
\end{lemma}

\begin{figure}
\centerline{\includegraphics{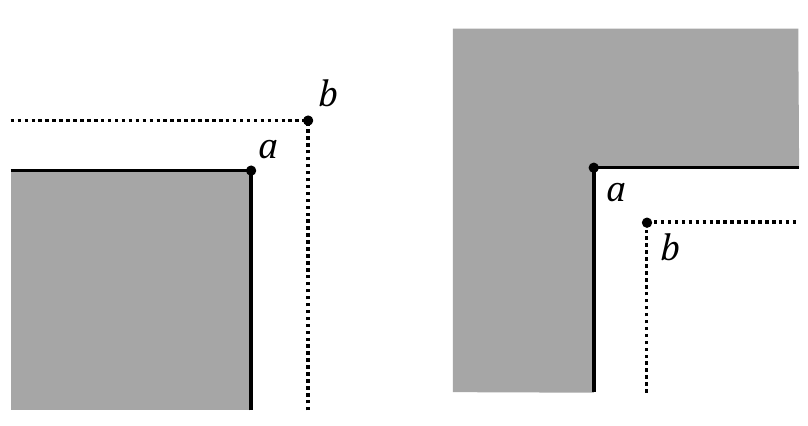}}
\caption{\label{fig_translate_st_halfsp}Translating a stair-halfspace outwards.}
\end{figure}

\begin{proof}
Let $p\in C_i(a)$ for some $i\in I$. We have $p_i \ge a_i >
b_i$ if $i\ge 1$, and $p_j \le a_j$ for each $i+1\le j\le d$.
We need to show that $p\in C_k(b)$ for some $k\in I$.

Let $k$ be the largest index such that $p\in C_k(b)$. Then
$p_k\ge b_k$ if $k\ge 1$, and $p_j<b_j$ for each $k+1\le j\le
d$.

If $k=i$ then $k\in I$ and we are done. Otherwise, we must have
$k>i$, or else we would have $p_i < b_i < a_i \le p_i$. But
$k>i$ implies that $b_k\le p_k\le a_k$. Since $b_k\neq a_k$, we
have $b_k<a_k$, which implies $k\in I$, as desired.
\end{proof}

\section{Proof of Theorem~\ref{thm_2d2_tight}}

In this section we prove that the stretched grid $\Gs$
satisfies Theorem~\ref{thm_2d2_tight}.

\begin{lemma}[Covering lemma]\label{lemma_exact_cover_lines}
Let $p, q$ be two points in $\R^d$, $d\ge 2$. Then there exists
a family $\mathcal H$ of $(d-1)(d+2)/2$ stair-halfspaces, each
one containing both $p$ and $q$, such that the stair-halfspaces
of $\mathcal H$ together cover $\R^d$ exactly $d-1$ times
(apart from the points on the boundary of the stair-halfspaces
of $\mathcal H$, which might be covered more times).
\end{lemma}

\begin{proof}
We proceed by induction on $d$. For the base case $d=2$, we
want to construct two stair-halfplanes containing $p$ and $q$
that cover $\R^2$ exactly once. Suppose without loss of
generality that $q_2 \ge p_2$. Let $a = (p_1, q_2)$. If $p_1
\le q_1$, then let $\mathcal H = \{ C_0(a) \cup C_2(a), C_1(a)
\}$; otherwise, let $\mathcal H = \{C_0(a), C_1(a) \cup C_2(a)
\}$ (see Figure~\ref{fig_covering_base}).

\begin{figure}
\centerline{\includegraphics{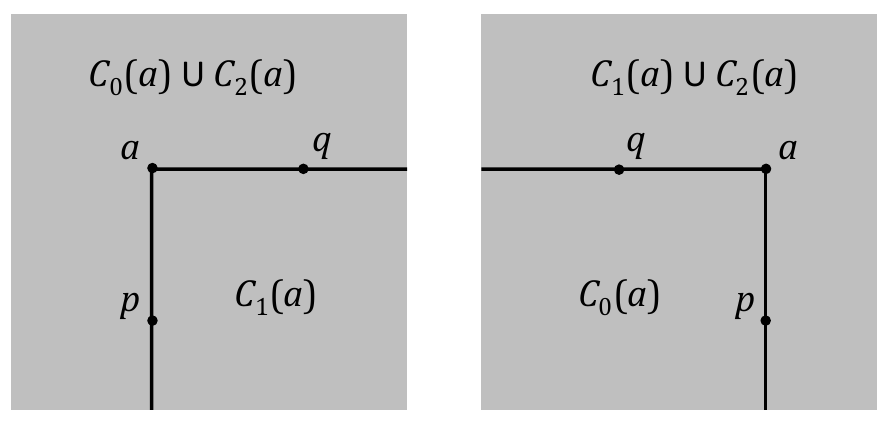}}
\caption{\label{fig_covering_base}The base case of the covering lemma.}
\end{figure}

Now assume $d\ge 3$. Let $\overline p$ and $\overline q$ denote
the vertical projection of $p$ and $q$ into $\R^{d-1}$,
respectively. By induction, let $\mathcal H'$ be a family of
$(d-2)(d+1)/2$ stair-halfspaces in $\R^{d-1}$ containing
$\overline p$ and $\overline q$ and covering $\R^{d-1}$ exactly
$d-2$ times.

Suppose without loss of generality that $p_d \le q_d$. As a
first step, construct the family of stair-halfspaces in $\R^d$
\begin{equation*}
\mathcal H^* = \{H \times (-\infty, q_d] : H\in \mathcal H'\}.
\end{equation*}
This corresponds to adding $q_d$ as the $d$-th coordinate to
the vertex of every stair-halfspace $H\in \mathcal H'$, and
then reinterpreting the components of $H$ as being
$d$-dimensional.

Note that $|\mathcal H^*| = (d-2)(d+1)/2$, that $p$ and $q$
belong to every stair-halfspace in $\mathcal H^*$, and that
$\mathcal H^*$ covers the \emph{lower part} of $\R^d$ (namely,
$\R^{d-1} \times (-\infty, q_d]$) exactly $d-2$ times, and it
does not cover the \emph{upper part} of $\R^d$ (namely,
$\R^{d-1} \times [q_d, \infty)$) at all.

Next, let $m$ be an index $0\le m\le d-1$ such that $\overline
q \in C_m(\overline p)$ (e.g., let $m\le d-1$ be the largest
index for which $p_m\le q_m$, or zero if no such index exists).
Let $a = (p_1, \ldots, p_{d-1}, q_d)$, and define the family of
stair-halfspaces
\begin{equation*}
\mathcal H^{**} = \{ C_i(a) \cup C_d(a) : 0\le i\le d-1, i\neq m \} \cup
\{C_m(a) \},
\end{equation*}
all having $a$ as vertex.

First, note that the stair-halfspaces of $\mathcal H^{**}$
cover the lower part of $\R^d$ exactly once (because each
component $C_0(a), \ldots, C_{d-1}(a)$ is present exactly
once), and they cover the upper part of $\R^d$ exactly $d-1$
times (because the component $C_d(a)$ is present $d-1$ times).

Furthermore, note that each stair-halfspace of $\mathcal
H^{**}$ contains both $p$ and $q$: The components $C_i(a)$,
$i\le d-1$ contain $p$; the component $C_d(a)$ contains $q$;
and the component $C_m(a)$ contains both $p$ and $q$, by the
choice of $m$.

Thus, the desired family of stair-halfspaces is $\mathcal H =
\mathcal H^* \cup \mathcal H^{**}$: It contains
\begin{equation*}
\frac{(d-2)(d+1)}{2} + d = \frac{(d-1)(d+2)}{2}
\end{equation*}
stair-halfspaces, and it covers $\R^d$ exactly $d-1$ times.
\end{proof}

\begin{remark}
The points $p$ and $q$ actually lie on the \emph{boundary} of
each stair-halfspace of $\mathcal H$. This can be seen by
recursively characterizing the boundary of a stair-halfspace (a
``stair-hyperplane"), and using induction.
\end{remark}

\begin{proof}[Proof of Theorem~\ref{thm_2d2_tight}]
Let $\Gs$ be the $n$-point stretched grid in $\R^d$, let $\B$
be its bounding box, and let $\ell$ be a line in $\R^d$. We
want to construct a Euclidean halfspace that contains $\ell$
and contains at most $2n/(d+2) + o(n)$ points of $\Gs$.

If $\ell$ does not intersect the interior of $\B$ then there is
nothing to do. Otherwise, let $p'$ and $q'$ be the intersection
points of $\ell$ with the boundary of $\B$, and let $p =
\pi(p')$, $q = \pi(q')$ be the corresponding points in the
boundary of $[0,1]^d$.

Let $\mathcal H$ be the family of stair-halfspaces guaranteed
by Lemma~\ref{lemma_exact_cover_lines} for the points $p$ and
$q$. By the pigeonhole principle, there must exist a
stair-halfspace $H\in\mathcal H$ such that $\vol(H\cap[0,1]^d)
\le 2/(d+2)$. Move the vertex $a$ of $H$ ``outwards" by
distance $1/(m-1)$ in each direction, so that $p$ and $q$ are
still contained in $H$ and are far enough from its boundary
(recall Lemma~\ref{lemma_translate_st_halfsp}). This increases
the volume of $H\cap[0,1]^d$ by only $o(1)$.

Now, the volume of $H\cap[0,1]^d$ closely approximates the
fraction of points of $\pi(\Gs)$ contained in $H$;
specifically, $\vol(H\cap[0,1]^d) = |\pi(\Gs) \cap H| / n \pm
O(n^{(d-1)/d})$. This is because $H\cap[0,1]^d$ is the union of
a constant number of axis-parallel boxes, and the claim is
clearly true for axis-parallel boxes.

Let $a' = \pi^{-1}(a)$, and let $H'$ be the stair-halfspace
combinatorially equivalent to $H$ having $a'$ as its vertex.
Then, $|\Gs \cap H'| = |\pi(\Gs) \cap H| \le 2n/(d+2) + o(n)$.
Furthermore, we have $p', q'\in H'$, and in fact, $p'$ and $q'$
are $1$-far from $a'$. Therefore, the Euclidean halfspace $H''$
promised by Lemma~\ref{lemma_stair-half_to_half} contains both
$p$ and $q$, and, like $H'$, it contains at most $2n/(d+2) +
o(n)$ points of $\Gs$.

It might still be possible that $H''$ does not contain all of
$\ell$, but this is easy to fix: The sets $\ell$ and $\B
\setminus H''$ are disjoint, and they are both convex.
Therefore, there exists a hyperplane $h$ that separates them.
Let $H'''$ be the halfspace bounded by $h$ that contains
$\ell$. Then $H''' \cap \B \subseteq H'' \cap \B$, so $H'''$
can only contain \emph{fewer} points of $\Gs$ than $H''$.
\end{proof}

\section{Generalization to $k$-flats}

We conjecture that the stretched grid $\Gs$ in fact gives a
tight upper bound of $(k+1)/(k+d+1)$ \emph{for all $k$} for the
leading constant in the centerflat conjecture.

We have a sketch of a proof. Its main ingredients are: (1) an
appropriate definition of \emph{stair-$k$-flats}, the
stair-convex equivalent of Euclidean $k$-flats; and (2) a
generalization of Lemma~\ref{lemma_exact_cover_lines} to the
effect that, for every stair-$k$-flat $f\subset \R^d$, there
exists a family $\mathcal H$ of
$\binom{d-1}{k}\frac{d+k+1}{k+1}$ stair-halfspaces, each one
containing $f$ in its boundary, and together covering $\R^d$
exactly $\binom{d-1}{k}$ times.

However, we have some problems formalizing the argument: We
have been unable to rigorously prove that our stair-flats are
indeed the ``limit case under $\pi$" of Euclidean $k$-flats,
and we have also been unable to deal with some degenerate
stair-flats.

For the interested reader, In Appendix~\ref{app_kflats} we
spell out the argument, pointing out the ``holes" that we still
have.

\paragraph{Acknowledgments.} Thanks to Ji\v r\'i Matou\v sek for some
insightful conversations on this and related topics, to Roman
Karasev for some useful correspondence, and to the anonymous
referee for giving a careful review and providing extensive
comments.

\bibliographystyle{plain}
\bibliography{upper_bds_centerlines_arxiv}

\appendix

\section{Generalization to $k$-flats: an incomplete argument}\label{app_kflats}

We define \emph{stair-$k$-flats}, which are the stair-convex analogues of
$k$-flats. Stair-$k$-flats in $\R^d$ are defined inductively on $k$ and $d$.
A stair-$k$-flat is always topologically equivalent to a regular $k$-flat;
this fact follows by induction, and it is inductively necessary for the
definition itself.

\begin{definition}
A \emph{stair-$0$-flat} is a point, and a \emph{stair-$d$-flat} in $\R^d$ is $\R^d$.
For $1\le k\le d-1$, a \emph{stair-$k$-flat}
$f$ in $\R^d$ has one of these three forms:
\begin{itemize}
\item (``Horizontal") $f = f' \times \{z\}$ for some stair-$k$-flat $f'$ in
$\R^{d-1}$ and some $z\in\R$.

\item (``Vertical") $f = f' \times (-\infty, \infty)$ for some stair-$(k-1)$-flat
$f'$ in $\R^{d-1}$.

\item (``Diagonal") Let $f'$ be a stair-$(k-1)$-flat in $R^{d-1}$, and let
$f''$ be a stair-$k$-flat in $R^{d-1}$ that contains $f'$. It follows by induction
that $f'$ is topologically equivalent to $\R^{k-1}$ and $f''$ is topologically
equivalent to $\R^k$; thus, $f'$ partitions $f''$ into two relatively closed
\emph{half-stair-$k$-flats} whose intersection equals $f'$. Let $h$ be one of
these halves. Then,
\begin{equation*}
f = \bigl( f'\times (-\infty, z] \bigr) \cup \bigl( h \times \{z\} \bigr),
\end{equation*}
for some $z\in\R$.
\end{itemize}
\end{definition}

\begin{figure}
\centerline{\includegraphics{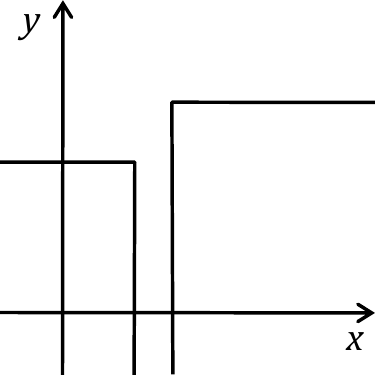}}
\caption{\label{fig_lines_R2}Stair-lines in the plane.}
\end{figure}

\begin{figure}
\centerline{\includegraphics{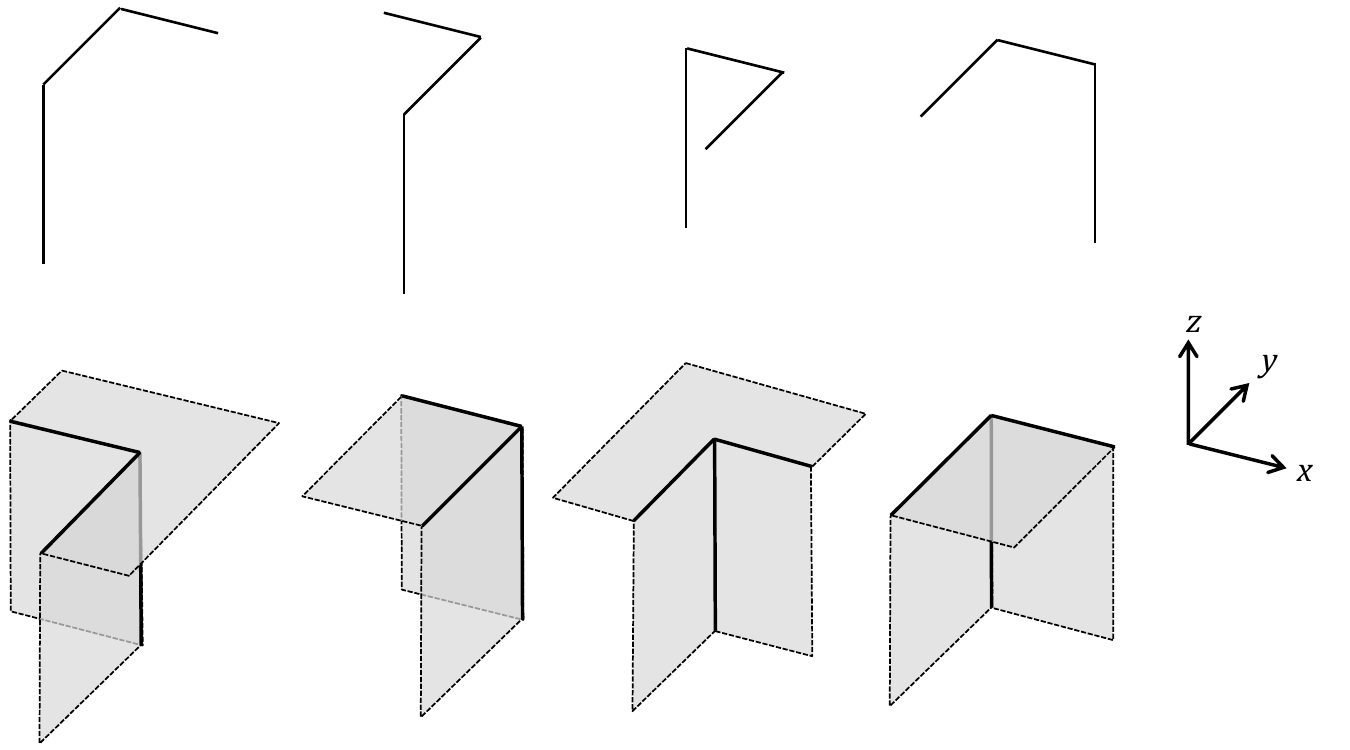}}
\caption{\label{fig_lines_planes_R3}Stair-lines and stair-planes in $\R^3$.}
\end{figure}

\begin{figure}
\centerline{\includegraphics{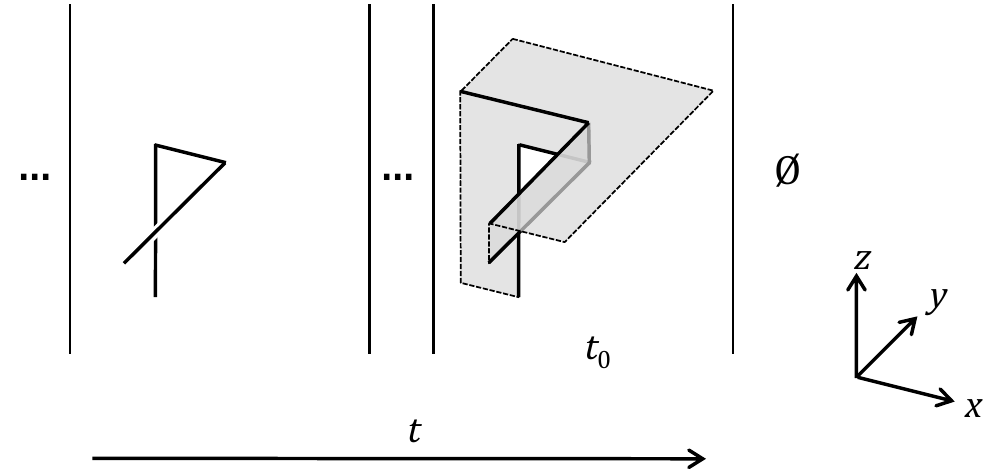}}
\caption{\label{fig_plane_R4}Example of a stair-plane in $\R^4$. The fourth coordinate
is $t$ (``time"). From $t=-\infty$ up to a certain time $t_0$ there exists a static
stair-line. At time $t_0$ a stair-halfplane bounded by this stair-line suddenly appears
for an instant, and then everything disappears.}
\end{figure}

See Figures~\ref{fig_lines_R2}, \ref{fig_lines_planes_R3}, and \ref{fig_plane_R4}
for some examples of stair-lines and stair-planes.

In a diagonal stair-flat $f$, the part $f' \times (-\infty, z]$ is
called the \emph{vertical part} of $f$, and the part $h\times \{z\}$ is called
the \emph{horizontal part} of $f$.

Diagonal stair-flats are the most general ones; the other ones can be considered
diagonal stair-flats for which either its horizontal or its vertical part has been moved
to infinity in some direction.

\begin{lemma}\label{lemma_half_convex}
\emph{(1)} Stair-flats are stair-convex. \emph{(2)} Closed half-stair-flats are stair-convex.
\end{lemma}

\begin{proof}
By induction. The first claim in $\R^d$ easily follows from the second
claim in $\R^{d-1}$ and monotonicity of slices.
And the second claim in $\R^d$ follows from the first claim in $\R^d$ as follows:

Let $h$ be a closed half-stair-$k$-flat in $\R^d$. Its relative boundary is
some stair-$(k-1)$-flat $f'$. Let $f$ be some stair-$k$-flat that contains $h$
(note that $f$ might not be unique). Let $a, b \in h$, and suppose for a contradiction
that $\sigma = \sigma(a, b)$ is not completely contained in $h$. Since $f$ is
stair-convex, $\sigma$ is completely contained in $f$. Thus, $\sigma$ must cross
$f'$ in at least two points $c$ and $d$ when going from $h$ to the other half of
$f$. But $f'$ is stair-convex, so the part of $\sigma$ between $c$ and $d$, which
equals $\sigma(c, d)$, never leaves $f'$. Contradiction.
\end{proof}

(We conjecture that every
stair-convex set in $\R^d$ that is topologically equivalent to $\R^k$ is a actually stair-$k$-flat.)

\begin{figure}
\centerline{\includegraphics{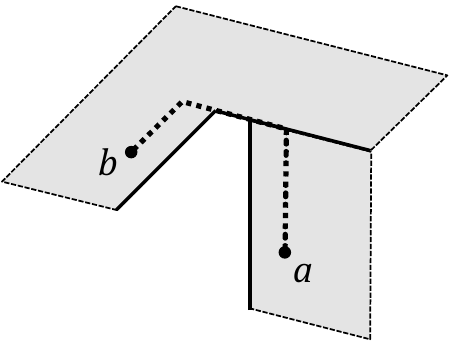}}
\caption{\label{fig_half_flat_stconv}A closed half-stair-flat is always stair-convex,
but an open one is not. The stair-path between $a$ and $b$ intersects
the relative boundary of the pictured half-stairplane.}
\end{figure}

Note that an \emph{open} half-stair-flat, unlike a \emph{closed} one, is not
necessarily stair-convex; see Figure~\ref{fig_half_flat_stconv}.

\subsection{Equivalence to Euclidean flats}

The following claim, for which we do not yet have a complete proof, states
that indeed stair-$k$-flats are the stair-convex equivalent of Euclidean
$k$-flats. Intuitively this means that, if $m$, the side of the stretched grid, is
very large, and if $f$ is a Euclidean $k$-flat that intersects $\B$, then $\pi(f \cap \B)$
looks almost like a stair-$k$-flat intersected with $\I^d$. Conversely, every
intersection of a stair-$k$-flat with $\I^d$ can be obtained this way.

Given two sets $f, g \subseteq \R^d$ and an integer $c\ge 0$, we say that $f$
and $g$ are \emph{$c$-close in $\B$} if every point of $f\cap\B$ is
$c$-close to a point of $g\cap \B$ and vice versa.
Normally, $f$ will be a flat and $g$ will be a stair-flat.

\begin{claim}\label{claim_transf}
For every Euclidean $k$-flat $f\subset\R^d$ there exists a stair-$k$-flat
$g\subset \R^d$ that is $c$-close to $f$ in $\B$ for some constant $c= c(d)$, and vice versa.
\end{claim}

\begin{proof}[``Proof"]
For the first direction, let $f$ be a Euclidean $k$-flat. If $f$ is vertical,
then the desired stair-$k$-flat $g$ is also vertical, and the claim follows by
induction on $d$. So suppose $f$ is not vertical.

Let $\overline{\B}$ be the bottom face of $\B$,
let $h_0 \subset \R^d$ be the horizontal hyperplane containing $\overline{\B}$,
and let $f' = f \cap h_0$. If $f' = \emptyset$ then $f$ is horizontal, so the desired
stair-$k$-flat $g$ is also horizontal, and it can again be constructed by induction on $d$.

Otherwise, $f'$ is a $(k-1)$-flat (which may or may not intersect $\overline{\B}$).
By induction, $f'$ is $c(d-1)$-close to some stair-$(k-1)$-flat $g'$ in $\overline{\B}$.

Next, let $\overline f$ be the vertical projection of $f$ into $h_0$. Again, by induction
$\overline f$ is $c(d-1)$-close to some stair-$k$-flat $g''$ in $\overline{\B}$.

Furthermore, since $f' \subset \overline f$, we know that $g'$ is very close to $g''$.\footnote{Here
is the problem: We would like to have $g'\subset g''$. One solution would be to ``snap" $g'$
into $g''$, whatever that means.}

Now, let $a$ and $b$ be two points on $f\cap \B$, and let $\overline a$ and $\overline b$ be
their vertical projections into $\overline{\B}$. One can verify that, if both $\overline a$ and $\overline b$
lie at Euclidean distance at least $1$ from $f'$, then $a$ and $b$ are $1$-close in last coordinate.
(This follows by a simple calculation involving ratios; compare to Lemma~\ref{lemma_property_Gs}.)
Intuitively, in $\pi(f\cap \B)$ all points have almost the same height, except for those that are
very close to $\pi(f'\cap \B)$ in the first $d-1$ coordinates, for which the height drops abruptly to zero.

Thus, the desired stair-$k$-flat $g$ is obtained as follows: Let $h$ be the ``correct" half-stair-flat
of $g''$ bounded by $g'$ (the half corresponding to the half of $f$ that goes \emph{up} in last coordinate).
Then pick an arbitrary point $a\in f$ with positive height such that its projection $\overline a$
lies in $\overline{\B}$ and has Euclidean distance greater than $1$
from $f'$. Let $z = a_d$, and let
\begin{equation*}
g = \bigl( g' \times (-\infty, z] \bigr) \cup \bigl(h \times \{z\} \bigr).
\end{equation*}
(The case where no such $a$ exists can also be taken care of; we omit the details.)

Let us briefly sketch the other direction: Let
\begin{equation*}
g = \bigl( g' \times (-\infty, z] \bigr) \cup \bigl(h \times \{z\} \bigr).
\end{equation*}
be a given stair-$k$-flat. Let $g''\subset \R^{d-1}$ be a full stair-$k$-flat that
contains $h$.
By induction, let $f'\subset h_0$ be the Euclidean $(k-1)$-flat that approximates $g'$ in
$\overline{\B}$,
and let $f''\subset h_0$ be the Euclidean $k$-flat that approximates $g''$ in $\overline{\B}$.

We know that $f'$ is ``close" to $f''$---not in the Euclidean sense, but in the sense of \emph{$c$-closeness}
based on the mapping $\pi$. As before, somehow ``snap" $f'$
into $f''$, getting a Euclidean $(k-1)$-flat $f''' \subset f''$.

Let $\overline a\in f''$ be a point that has Euclidean distance at least $1$ from $f'''$; elevate
$\overline a$ vertically to height $z$, getting point $a$; and finally let $f$ (the desired
Euclidean $k$-flat) be the affine hull of $f'''$ and $a$.
\end{proof}

\subsection{The generalized covering lemma}

\begin{lemma}[Generalized covering lemma]\label{lemma_exact_cover}
Let $f$ be a stair-$k$-flat in $\R^d$. Then, there exists a family $\mathcal H$ of
$\binom{d-1}{k}\frac{d+k+1}{k+1}$ closed stair-halfspaces, each one containing
$f$ in its boundary, and together covering $\R^d$ exactly $\binom{d-1}{k}$ times
(apart from the points lying on the boundary of the stair-halfspaces of $\mathcal H$,
which might be covered more times).
\end{lemma}

\begin{proof}
We construct $\mathcal H$ by induction on $k$ and $d$.

Let $\Gamma = \Gamma_{k,d} = \binom{d-1}{k}\frac{d+k+1}{k+1}$ denote the desired
number of stair-halfspaces, and let $\Delta = \Delta_{k,d} = \binom{d-1}{k}$
denote the number of times space should be covered.

When $k=0$, $f$ consists of a single point $a$, and we have $\Gamma = d+1$,
$\Delta = 1$. In this case we let $H = \{C_0(a), \ldots, C_d(a)\}$, and we
have $d+1$ stair-halfspaces, all containing $a$ in their boundary, and together covering
space exactly once, as required.

Now suppose $k\ge 1$, and assume that $f$ is a diagonal stair-$k$-flat (the other types
of stair-flats are degeneracies, as mentioned above). Thus, $f$ has the form
\begin{equation*}
f = \bigl( f' \times (-\infty, z] \bigr) \cup \bigl( h \times \{z\} \bigr),
\end{equation*}
for some stair-$(k-1)$-flat $f'$ and some half-stair-$k$-flat $h$, both in
$\R^{d-1}$, such that $f'$ is the relative boundary of $h$.
Let $f''$ be a full stair-$k$-flat in $\R^{d-1}$ containing $h$
(there might be more than one way to ``complete" $h$ into a stair-flat).

Let $\Gamma' = \Gamma_{k-1,d-1}$, $\Delta' = \Delta_{k-1, d-1}$, $\Gamma''
= \Gamma_{k,d-1}$, $\Delta'' = \Delta_{k,d-1}$. By induction, we can construct
a family $\mathcal H'$ of $\Gamma'$ stair-halfspaces in $\R^{d-1}$, all containing $f'$
in their boundary and covering $\R^{d-1}$ exactly $\Delta'$ times, and a family
$\mathcal H''$ of $\Gamma''$ stair-halfspaces in $\R^{d-1}$, all containing $f''$ in
their boundary and covering $\R^{d-1}$ exactly $\Delta''$ times.

Also note the following identities:
\begin{align}
\Gamma &= \Gamma' + \Gamma'', \nonumber \\
\Delta &= \Delta' + \Delta'', \nonumber \\
\Gamma' &= \Delta + \Delta'. \label{eq_GDD}
\end{align}

We will construct our desired family $\mathcal H$ of stair-halfspaces as
$\mathcal H = \mathcal H_1 \cup \mathcal H_2$,
with $|\mathcal H_1| = |\mathcal H'| = \Gamma'$ and $|\mathcal H_2| = |\mathcal H''| = \Gamma''$.

Let us start by constructing $\mathcal H_2$, which is easier. We let
\begin{equation*}
\mathcal H_2 = \{ H'' \times (-\infty, z] : H'' \in \mathcal H''\}
\end{equation*}
(namely, we ``extrude" each halfspace of $\mathcal H''$ in the $d$-th direction from
$-\infty$ to $z$).

Let $H = H''\times (-\infty, z]$ be a stair-halfspace in $\mathcal H_2$. Note that $H$
contains all of $f$, as required. Furthermore, the boundary of $H$ is
\begin{equation*}
\partial H = \bigl( \partial H'' \times (-\infty, z] \bigr) \cup
\bigl( H'' \times \{z\} \bigr),
\end{equation*}
so $f$ is actually contained in the boundary of $H$, as required.

Note that the stair-halfspaces of $\mathcal H_2$
cover the ``lower part" of $\R^d$ (meaning, $\R^{d-1} \times (-\infty, z]$) exactly $\Delta''$
times, and they do not cover the ``upper part" of $\R^d$ (meaning, $\R^{d-1} \times
[z,\infty)$) at all.

We now construct $\mathcal H_1$.
Let us first take a more careful look at the stair-halfspaces of $\mathcal H'$. Recall that the
stair-halfspaces of $\mathcal H'$ were only ``designed" to contain $f'$, but not $h$.

\begin{lemma}\label{lemma_H3types}
The family $\mathcal H'$ can be partitioned into $\mathcal H' = \Ha \cup \Hb \cup \Hc$, such that:
\begin{itemize}
\item $h$ intersects the interior of every $H\in\Ha$;
\item $h$ is contained in the boundary of every $H\in\Hb$; and
\item $h$ is not contained in any $H\in\Hc$.
\end{itemize}
\end{lemma}

\begin{proof}
We have to show that, if $h$ intersects the interior of some $H\in\mathcal H'$, then $h\subset H$.
Recall that the relative boundary of $h$, namely $f'$, lies on the \emph{boundary} of every $H\in\mathcal H'$.

So supppose for a contradiction that
there exists a half-stair-$k$-flat $h\subset \R^d$ and there exists a stair-halfspace $H\subset
\R^d$ such that $\partial H$ contains the relative boundary of $h$, and such that $h$ intersects
\emph{both} the interior of $H$ \emph{and} $\R^d\setminus H$. Then, by Claim~\ref{claim_transf}
we could construct a similar configuration with a \emph{Euclidean} half-$k$-flat and a \emph{Euclidean}
halfspace. But that is clearly impossible.\footnote{We would like a proof of Lemma~\ref{lemma_H3types}
that uses only the notions of stair-convexity and does not invoke Claim~\ref{claim_transf}; but
we have not found such a proof.}
\end{proof}

\begin{claim}\label{claim_D_prime}
We have $|\Ha| \le \Delta' \le |\Ha \cup \Hb|$ and $ |\Hc| \le \Delta \le |\Hb \cup \Hc|$.
\end{claim}

\begin{proof}[``Proof"]
Let $S = \R^d \setminus \bigcup_{H\in\mathcal H'} \partial H$ be the set of all points not lying
on the boundary of any stair-halfspace of $\mathcal H'$.

Supose first that our half-stair-flat $h$ intersects $S$ (so $h$ is, in some sense, ``generic"),
and let $a$ be a point in $h\cap S$. Then, by Lemma~\ref{lemma_H3types},
for every $H\in\mathcal H'$ we have $h\subset H$ if and
only if $a\in H$. By the construction of $\mathcal H'$ we know that it contains exactly $\Delta'$
halfspaces that satisfy this latter property; therefore, $|\Ha| = \Delta'$ and $\Hb = \emptyset$.
Then equation (\ref{eq_GDD}) implies that $|\Hc| = \Delta$, and we are done.

If $h$ does not intersect $S$, then we apply a limit argument: $h$ is
arbitrarily close to a half-stair-$k$-flat $h'$, having the same relative boundary as $h$,
such that $h'$ \emph{does} intersect $S$.\footnote{This would need a proof, of course.}
Define the sets $\Ha, \Hb, \Hc$ of Lemma~\ref{lemma_H3types}
for $h'$, and continuously
``rotate" $h'$ until it matches $h$. At the beginning,
we have $|\Ha| = \Delta'$, $\Hb=\emptyset$; and the only combinatorial changes that can occur
involve moving \emph{into} the boundary of stair-flats $H\in \mathcal H'$. In other words, elements
can only move from $\Ha$ or $\Hc$ into $\Hb$.
\end{proof}

Now, partition $\mathcal H'$ into $\mathcal H' = \Hup \cup \Hdown$ such that $\Hc \subseteq \Hup$
and $\Ha\subseteq \Hdown$, and such that $|\Hup| = \Delta$
and $|\Hdown| = \Delta'$.

Then our desired family $\mathcal H_1$ is
\begin{align*}
\mathcal H_1 &= \bigl\{ \bigl( H \times (-\infty, z] \bigr) \cup \bigl(\R^{d-1} \times [z, \infty)
\bigr) : H\in \Hup \bigr\} \\
& \cup \bigl\{ H \times (-\infty, z] : H \in \Hdown \bigr\}
\end{align*}
(extruding every stair-halfsace in the $d$-th direction as before, and adding the $d$-th component
only to the stair-halfspaces of $\Hup$).

The condition $\Hc\subseteq \Hup$ implies that every stair-halfspace of $\mathcal H_1$ contains $h\times \{z\}$,
and thus all of $f$.

Moreover, for a stair-halfspace $H\in\Hup$, the boundary of the corresponding stair-halfspace $H'\in
\mathcal H_1$ is
\begin{equation*}
\partial H' = \bigl( \partial H \times (-\infty, z] \bigr) \cup \bigl( (\R^{d-1}\setminus H) \times
\{z\} \bigr),
\end{equation*}
so actually $f\in\partial H'$.

Finally, note that $\mathcal H_1$ covers the ``lower part" of $\R^d$ exactly $\Delta'$ times and
the ``upper part" of $\R^d$ exactly $\Delta$ times. Hence, $\mathcal H = \mathcal H_1 \cup \mathcal
H_2$ is our desired family.
\end{proof}

\begin{figure}
\centerline{\includegraphics{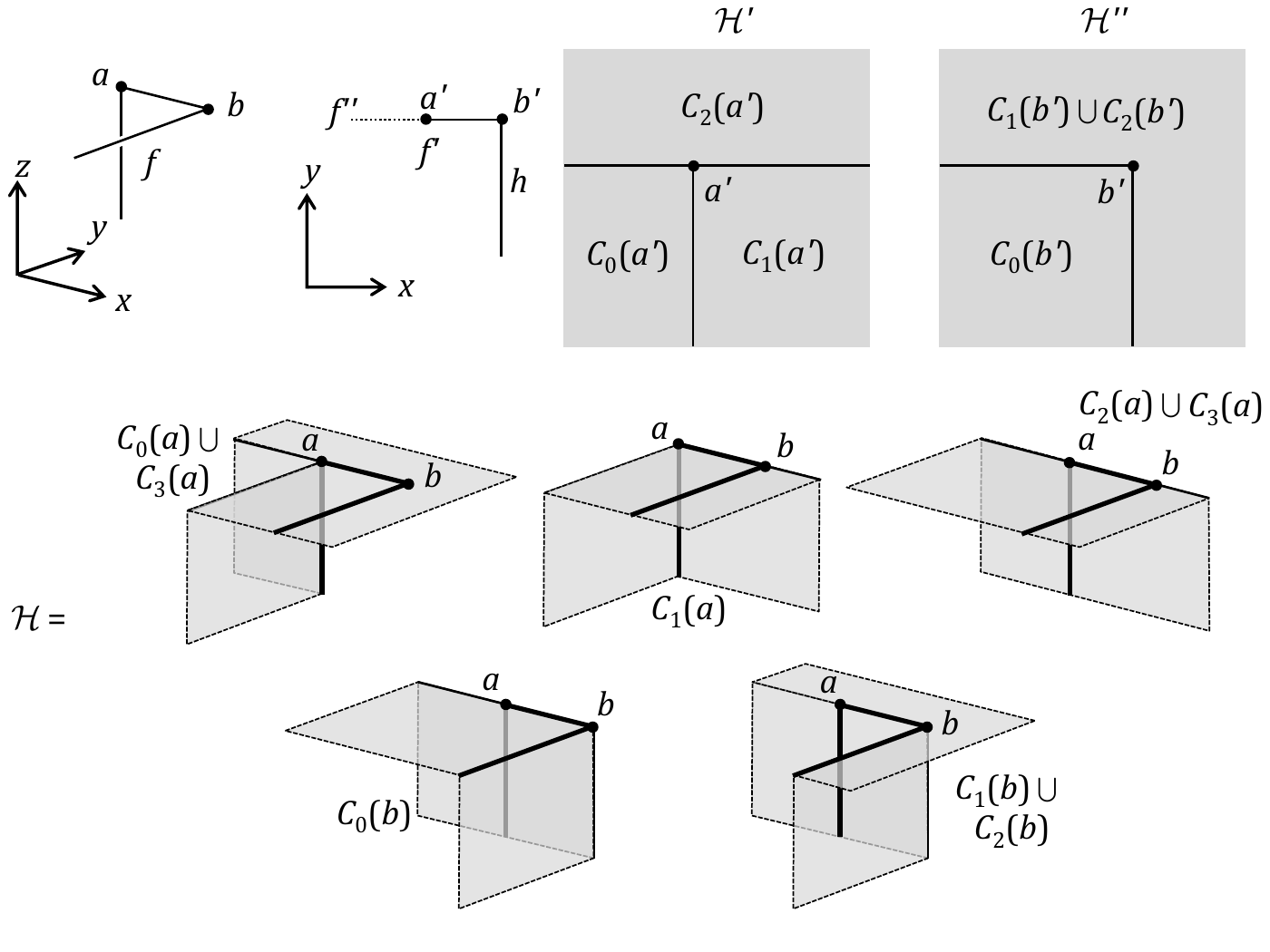}}
\caption{\label{fig_construct_H}An example of constructing $\mathcal H$ from $\mathcal H'$ and $\mathcal H''$.}
\end{figure}

Figure~\ref{fig_construct_H} shows an example of the construction of
Lemma~\ref{lemma_exact_cover} for a stair-line $f$ (so $k=1$) in $\R^3$. The stair-line $f$
is shown at the top left. To its right are shown the two-dimensional components from
which $f$ is made: The stair-point $f'$ and the stair-ray $h$. The stair-line $f''$ that
contains $h$ is also shown.

Next are shown the three stair-halfplanes of $\mathcal H'$, which contain $f'$ in their boundary and together
cover the plane exactly once. We have $\mathcal H' = \{ C_0(a'), C_1(a'), C_2(a') \}$.

Next are the two stair-halfplanes of $\mathcal H''$, which contain $f''$ in their boundary and together cover
the plane exactly once. We have $\mathcal H'' = \{ C_0(b'), C_1(b') \cup C_2(b') \}$.

Finally are shown the five stair-halfspaces of the desired set $\mathcal H$, which
contain $f$ in their boundary and together cover space exactly twice. We have
\begin{equation*}
\mathcal H'' = \{ C_0(a) \cup C_3(a), C_1(a), C_2(a) \cup C_3(a), C_0(b), C_1(b) \cup C_2(b) \}.
\end{equation*}

The adventurous reader might want to try to construct the seven stair-halfspaces that cover $\R^4$ three
times for the stair-plane in Figure~\ref{fig_plane_R4}.

\end{document}